%% file: main.tex
\documentclass[12pt]{article}

\usepackage{natbib}
 \bibpunct{(}{)}{;}{a}{,}{,}
\usepackage[top=1in, bottom=1in, left=1in, right=1in]{geometry}
\usepackage[pdftex,colorlinks=true,linkcolor=blue,citecolor=blue,urlcolor=blue,bookmarks=false,pdfpagemode=None]{hyperref}
\usepackage{url}

\usepackage{amsmath, amsthm, amssymb}
\usepackage{booktabs} % For formal tables
\usepackage[ruled]{algorithm2e} % For algorithms
\usepackage{bbm}
\usepackage{graphicx}
\usepackage{textcomp}
\usepackage{setspace}
\newtheorem{lemma}{Lemma}
\newtheorem{theorem}{Theorem}
\newtheorem{definition}{Definition}
\newtheorem{proposition}{Proposition}
\newtheorem{property}{Property}
\newcommand{\E}{\mathbb{E}}
\newcommand{\calC}{\mathcal{C}}
\newcommand{\N}{\mathcal{N}}
\newcommand{\calP}{\mathcal{P}}
\newcommand{\Y}{\mathbf{Y}}
\newcommand{\W}{\mathbf{W}}
\newcommand{\Z}{\mathbf{Z}}
\DeclareMathOperator*{\argmax}{arg\,max}
\begin{document}
\pagestyle{empty}

\title{Optimizing cluster-based randomized experiments\\ under a monotonicity assumption%
\protect\thanks{%Alexander M. Franks is an Assistant Professor of Statistics at the University of California, Santa Barbara (\href{mailto:afranks@pstat.ucsb.edu}{afranks@pstat.ucsb.edu}) Edoardo M.~Airoldi is an Associate Professor of Statistics at Harvard University (\href{mailto:airoldi@fas.harvard.edu}{airoldi@fas.harvard.edu}). Donald B. Rubin is the John L. Loeb Professor of Statistics at Harvard University (\href{mailto:dbrubin@fas.harvard.edu}{dbrubin@fas.harvard.edu}).
This work was partially supported 
 by the National Science Foundation under grants 
  CAREER IIS-1149662 and IIS-1409177,
 by the National Institute of Health under grant 
  R01 GM-096193, 
 by the Office of Naval Research under grants 
  YIP N00014-14-1-0485 and N00014-17-1-2131, 
 by a Shutzer Fellowship to Edoardo M.~Airoldi,
 and by a Siebel Scholarship to Jean Pouget-Abadie.
}}
 
\author{Jean Pouget-Abadie$^\dagger$, ~David C. Parkes$^\dagger$, ~Vahab
Mirrokni$^\ddagger$ \and Edoardo M. Airoldi$^\dagger$}
\date{$^\dagger$Harvard University \quad $^\ddagger$Google}

\maketitle
\thispagestyle{empty}

\newpage
\begin{abstract}
Cluster-based randomized experiments are popular designs for mitigating the
bias of standard estimators when interference is present and classical causal
inference and experimental design assumptions (such as SUTVA or ITR) do not
hold. Without an exact knowledge of the interference structure, it can be
challenging to understand which partitioning of the experimental units is
optimal to minimize the estimation bias. In the paper, we introduce a
monotonicity condition under which a novel two-stage experimental design allows
us to determine which of two cluster-based designs yields the least biased
estimator. We then consider the setting of online advertising auctions and show
that reserve price experiments verify the monotonicity condition and the
proposed framework and methodology applies. We validate our findings on an
advertising auction dataset.

\vfill
\end{abstract}

\newpage
% \setcounter{tocdepth}{4}
% 1 Section
% 2 Subsection
% 3 Subsubsection
% 4 Paragraph
% 5 Subparagraph
%\footnotesize
%\renewcommand\baselinestretch{0.1}
\tableofcontents

\newpage
\pagestyle{plain}
\setcounter{page}{1}
\input{intro.tex}
\input{theory.tex}

%\input{sections/empirical.tex}
\input{application.tex}

\input{experimental.tex}
\input{conclusion.tex}

\appendix
\input{proofs.tex}

\bibliographystyle{plainnat}
\bibliography{main}

\end{document}

%% file: intro.tex
\section{Introduction}

Randomized experiments --- or A/B tests --- are at the core of many product
decisions at large technology companies. Under the commonly assumed Stable Unit
Treatment Value Assumption (SUTVA), these A/B tests serve to estimate unbiasedly
the effect of assigning all units to a particular intervention over an
alternative condition \citep{imbens2015causal}. The SUTVA assumption is one of no
interference between units: a unit's outcome in the experiment does not depend
on the treatment assignment of any other unit.

In many A/B tests however, this assumption is not tenable. Consider an
intervention on a user of a messaging platform: the (potential) resulting change
in her behavior (e.g. increase in time spent on the platform, in
number of messages sent, a decrease in response time) would affect the
friends on the platform she chooses to communicate with. The same cascading
phenomenon can also occur in more subtle ways in a social feed setting. Changes
to a feed ranking algorithm, and the resulting behavioral changes (e.g. a higher
click-through rate, feedback, or interaction time with the content on the feed)
will invariably affect the content on that unit's friends' social
feeds~\citep{eckles2016estimating, gui2015network}.

In particular, the same is true in an advertiser auction setting, where
modifications to the ecosystem can impact auctions and bidders not originally
assigned to the intervention~\citep{basse2016randomization}. Suppose that one
bidder changes her strategy as a result of being assigned to a higher reserve
price, or her usual bid no longer meets the reserve. The bidders she competes
with now face a different bid distribution --- the auction is now more
competitive if she increases her bid to meet the new reserve, or less
competitive if she fails to meet the reserve.  These bidders might react to this
new bid distribution by updating their own bidding strategy, even though they
were not originally assigned to the intervention. This effect could potentially
affect the other auctions they participate in.

When SUTVA does not hold, we say there is \emph{interference} between units, and
many fundamental results of the causal inference literature no longer hold. For
example, the difference-in-means estimator under a completely randomized
assignment is no longer unbiased~\citep{imbens2015causal}. When the estimand is
the difference of outcomes under two extreme assignments --- one assigning all
units to the intervention, and the other assigning none --- a common approach to
mitigating the bias of standard estimators in the face of interference is to run
cluster-based randomized designs~\citep{ugander2013graph, walker2014design,
eckles2017design}. These randomized designs group assign units
to treatment or control in groups to limit the amount of interaction between
different treatment buckets.

If it can be shown that there is no interaction across treatment buckets, we
recover many of the results stated under SUTVA. In practice, however, such a
grouping of units may not exist and A/B test practitioners often settle
to find the best possible partitioning. The problem is often formulated as the
balanced partitioning of a weighted graph on the experimental
units, where an edge is drawn between two units that are liable to interfere
with one another. This is a challenging task, both algorithmically and
empirically: clustering a graph into balanced partitions is known to be NP-hard,
even if we tolerate some unevenness between
partitions~\citep{andreev2006balanced}; furthermore,
the correct graph representation of the interference mechanism is not always
clear.

While the literature on finding balanced partitioning of weighted graphs and
analysing cluster-based randomized designs is
extensive~\citep{middleton2011unbiased, donner2004pitfalls, eckles2017design},
there are relatively few prior works that tackle the following question: can we
determine which of two balanced partitionings produces less biased estimates of
the total treatment effect, without assuming the exact structure of interference
is known?  The objective of this paper is to show we can in fact identify the
better of two clusterings through experimentation under an assumption on
the interference mechanism, which we call {\em monotonicity}.

Even when the exact structure of interference is not known, monotonicity can
established under a  theoretical model. For example, some interference
mechanisms are {\em self-exciting} --- if assigning any unit to the intervention
will boost the outcomes of any neighboring units.  Examples range from
vaccination campaigns to  social feed ranking algorithms. In both cases, the
units in the vicinity of a unit assigned to the intervention tend to benefit
over those surrounded by units in the control bucket. Interference mechanisms
that exhibit this self-exciting property are a particular example of monotone
mechanisms (cf.  Section~\ref{sec:monotonicity}).  When monotonicity holds, we
show that it is feasible to compare two balanced partitionings of the
experimental units by running a straightforward modification of an
experiment-of-experiments design~\citep{saveski2017detecting,
pouget2017testing}.

We make the following contributions: we present an experiment-of-experiments
design for comparing cluster-based randomized designs. We define a monotonicity
assumption under which we can determine which clustering induces the least
biased estimates of the total treatment effect using this comparative design.
While our technique applies to the general problem of experimental design under
interference with a monotonicity assumption, we prove that pricing
experiments\footnote{While pricing experiments are done in the context of ad
exchanges~\citep{AdExchange}, we note that our paper is a theoretical study of
the subject and does not include any real treatments of ad campaigns.} in the
context of ad exchanges are monotone, and thus our framework applies to this
illustrative example. Finally, we report an empirical simulation study of our
algorithms for a publicly-available dataset for online ads.

In Section~\ref{sec:theory}, we establish the theoretical framework  by defining
the monotonocity assumption, describing the suggested experiment-of-experiments
design, and proposing a test for interpreting its results. In
Section~\ref{sec:application}, we explain how this framework can be applied to
a real-world setting, by showing that reserve-price experiments on advertising
auctions are monotone. Finally, we validate these findings on a Yahoo! ad
auction dataset in Section~\ref{sec:experimental}.

%% file: theory.tex
\section{Theory}
\label{sec:theory}

In this section, we set the notation for the estimand, estimates, and
cluster-based randomized designs that we study. We then define the
monotonicity assumption, introduce our experiment-of-experiments design, and
suggest an approach to analysing its results.

\subsection{Cluster-based randomized designs}

Let $N$ be the number of experimental units, let vector $\Y$ denote the
outcome metric of interest, and let  vector $\Z$ denote the assignment of
units to treatment $(Z_i = 1)$ or control $(Z_i = 0)$.  Recall that under the
potential outcomes framework, $\Y(\Z)$ denotes the potential outcomes of the $N$
units under assignment $\Z$. Under the Stable Unit Treatment Value Assumption
(SUTVA), this simplifies to ${(Y_i(Z_i))}_1^N$. The estimand of interest here
is the {\em Total Treatment Effect} (TTE), defined as the difference of outcomes
between one assignment assigning all units to treatment, and another assigning
none:
\begin{equation}
\label{eq:tte}
TTE = \frac{1}{N} \sum_{i = 1}^N Y_i(\Z = \vec 1) - Y_i(\Z = \vec 0)
\end{equation}

A completely randomized (CR) design assigns $N_T$ units chosen completely at
random to treatment and the remaining $N_C = N - N_T$ units to control.  A
clustering $\calC$ is a partition of the $N$ experimental units into $M$
clusters. A {\em cluster-based randomized} (CBR) design is a randomized assignment of
units to treatment and control at the cluster level: if cluster $j$ is assigned to
treatment (resp.~control), then all units in cluster $j$ are assigned to
treatment (resp.~control).  We will use the notation $\E_{\Z \sim \calC}[X]$ to
denote the expected value of estimator $X$ under a $\calC$-cluster-based
randomized design. Recall that $\Z \sim \calC$ represents the assignment of
units to treatment and control, resulting from assigning the \emph{clusters} of
$\calC$ uniformly at random to treatment or control.

Let $M_T$ (resp. $M_C$) be the number of clusters assigned to treatment
(resp.~control).  Let $z \in {\{0, 1\}}^{M}$ be the assignment vector over
\emph{clusters}, where $M = M_T + M_C$.  In practice, we will use the
Horvitz-Thompson (HT) estimator, defined below:
\begin{equation}
  \label{eq:HT}
    \hat \tau = \frac{M}{N} \left( \frac{1}{M_T} \sum_{j=1}^M z_j \sum_{i \in
    \calC_j} Y_i(\Z) - \frac{1}{M_C} \sum_{j=1}^M (1 - z_j) \sum_{i \in \calC_j}
Y_i(\Z)  \right)
\end{equation}

Under SUTVA, the HT estimator is an unbiased estimator of the total treatment
effect under any $\calC$-CBR assignment~\citep{middleton2011unbiased}:
\begin{equation*}
    \E_{\Z \sim \calC}[\hat \tau] = TTE
\end{equation*}

When SUTVA does not hold, this property is no longer guaranteed, and $\hat{\tau}$
may be biased. Our objective is to minimize the bias, defined below, with
respect to the clustering, without assuming any explicit knowledge of the
interference mechanism or the value of the estimand $TTE$:
\begin{equation}
    \label{eq:objective}
    \min_{\calC} | \E_{\Z \sim \calC}[\hat \tau] - TTE|
\end{equation}

\subsection{A monotonicity assumption}
\label{sec:monotonicity}
Choosing the partitioning of our experimental units in a way that minimizes the
bias of our estimators~(cf.~Eq.~\ref{eq:objective}) when running a cluster-based
experiment is a difficult task: without the ground truth, we cannot observe the
bias directly.  However, under a specific monotonicity property--- common to
many randomized experiments ---the task of choosing the better of two
clusterings becomes straightforward.
\begin{definition}
\label{def:one-sided}
Let $\calP$ be the set of all possible clusterings of our $N$ units. For a subset
$\calP' \subset \calP$ of possible clusterings, we say that the interference
model is {\em $\calP'$-increasing} if and only if
\begin{equation*}
    \forall \calC \in \calP',~\E_{\Z \sim \calC}[\hat \tau] \leq \tau,
\end{equation*}
and it is {\em $\calP'$-decreasing} if and only if
\begin{equation*}
    \forall \calC \in \calP',~\E_{\Z \sim \calC}[\hat \tau] \geq \tau
\end{equation*}
A $\calP'$-\emph{monotone} model is one that is either $\calP'$-increasing or
  $\calP'$-decreasing.
\end{definition}

A monotone model is one for which the expectation of the HT
estimator $\hat \tau$ is either always a lower bound or always an
upper-bound of the estimand under any $\calC$-CBR design for $\calC \in \calP'$.
It is sufficient for $\calP'$ to contain the partitions we wish
to compare: we do not have to prove monotonicity beyond those partitions.
Before delving into examples of monotone interference mechanisms, we introduce
the following proposition, which highlights why monotonicity is useful 
for reasoning about bias.
\begin{proposition}
\label{prop:usefulness}
If the interference model is $\calP'$-increasing, then for all $\calC_1, \calC_2
  \in \calP'^2$, it holds that
\begin{equation*}
    \E_{\Z \sim \calC_1}[\hat \tau] \leq \E_{\Z \sim \calC_2}[\hat \tau]
    \implies |\E_{\Z\sim \calC_1}[\hat \tau] - \tau| \geq |\E_{\Z \sim
    \calC_2}[\hat \tau] - \tau|
\end{equation*}
If the interference model is $\calP'$-decreasing, then for all $\calC_1,
\calC_2 \in \calP'^2$, it holds that
\begin{equation*}
    \E_{\Z \sim \calC_1}[\hat \tau] \leq \E_{\Z \sim \calC_2}[\hat \tau]
    \implies |\E_{\Z\sim \calC_1}[\hat \tau] - \tau| \geq |\E_{\Z \sim
    \calC_2}[\hat \tau] - \tau|
\end{equation*}
\end{proposition}

Proposition~\ref{prop:usefulness} is a simple consequence of
Definition~\ref{def:one-sided}: if we know that two cluster-based estimates are
both lower bounds of the estimand, then the greater of the two must be less
biased. The same reasoning applies if they both upper-bound the estimand. It is
sufficient to compare the expectation of our estimators to determine which is
less biased.

The crux of our framework therefore relies on reasoning about monotonicity.
Many commonly studied parametric models of interference are in fact monotone.
Consider the following {\em linear model of interference} (e.g. studied
in~\citep{eckles2017design}):
\begin{equation}
    \label{eq:linear}
    Y_i(\Z) = \alpha_i + \beta_i Z_i + \gamma \rho_i + \epsilon_i,
\end{equation}
where for all $i$, $(\alpha_i, \beta_i, \gamma) \in \mathbb{R}^3$, $\epsilon_i
\sim \mathcal{N}(0, 1)$ is independent of $\rho_i$, and $\rho_i =
\frac{1}{|\N_i|} \sum_{j \in \mathcal{N}_i} Z_j$ is the proportion of $i$'s
neighborhood $\mathcal{N}_i$ that is treated. This expresses each unit's outcome
as a linear function of a fixed effect, a heterogeneous treatment effect, and a
network effect proportional to the fraction of my neighborhood that is treated.
As shown in the following proposition, this is  monotone.
\begin{proposition}\label{prop:simple_linear_monotone}
    For all $\calC \in \calP$, let $\theta_\calC = \frac{1}{N} \sum_i
    \frac{|\N_i \cap \calC(i)|}{|\N_i|}$ be the average proportion of a unit
    $i$'s neighborhood $\N_i$ included in its assigned cluster $\calC(i)$. Then,
    \begin{equation*}
        \tau - \E_{\Z \sim \calC}[\hat \tau] = \frac{\gamma M}{M-1} (1 -
    \theta_\calC)
    \end{equation*}
    It follows that if $\gamma \geq 0$, the interference model is
    $\calP$-increasing, otherwise it is $\calP$-decreasing.
\end{proposition}

We can also extend the above for heterogeneous network effect
parameters $\gamma_i$. A proof can be found in Section~\ref{sec:proofs}.
\begin{proposition}\label{prop:linear_monotone}
Let $\theta_{\calC, i} = \frac{|\N_i \cap \calC(i)|}{|\N_i|}$.  For all $\calC
  \in \calP$,
\begin{equation*}
  \tau - \E_{\Z \sim \calC}[\hat \tau] = \frac{M}{N (M-1)}\sum_i
    \gamma_i (1 - \theta_{\calC, i})
\end{equation*}
It follows that if $\sum_i \gamma_i(1- \theta_i) \geq 0$, then the interference
model is $\calP$-increasing, otherwise it is $\calP$-decreasing.
\end{proposition}
It follows that if $\gamma_i \geq 0, \forall i$, then the interference mechanism
is $\calP$-increasing, and if $\gamma_i \leq 0, \forall i$, then it is
$\calP$-decreasing. If the sign of $\gamma_i$ is not consistent, then the
monotonicity depends on the clustering: if all units with a given
sign are perfectly clustered $(\theta_{C, i} = 1)$, e.g.~all units with
$\gamma_i \geq 0$, then the mechanism is once again monotone.

More sophisticated interference mechanisms, without an immediate parametric
form, are also monotone. For example, we show that the
interference mechanism present in reserve price experiments in an advertiser
auction setting is monotone (under certain conditions). See
Section~\ref{sec:application} for more details. For these complex interference
mechanisms, it can also be easier to establish the following sufficient (but not
necessary) condition:
\begin{proposition}
\label{prop:more}
We say an interference mechanism verifies the \emph{self-excitation property}
  for a set of partitions~$\calP'$, if for all units $i$ and partitions $\calC
  \in \calP'$,
\begin{align*}
  & \E_{\Z \sim \calC}[ Y_i(\Z) : Z_i = 0] \geq Y_i(\vec 0) \\
    & \E_{\Z \sim \calC}[ Y_i(\Z) : Z_i = 1] \leq Y_i(\vec 1)
\end{align*}
A $\calP'$-self-exciting process is $\calP'$-increasing. A {\em self-deexciting
  mechanism}, with flipped inequalities, is $\calP'$-decreasing.
\end{proposition}

The proof is included in Section~\ref{sec:proofs}. The two
inequalities capture the following phenomenon: conditioned on my treatment
status, if my outcome is greatest when my neighborhood is entirely in treatment,
and lowest when my neighborhood is entirely in control, then an experiment
always under-estimates the true treatment effect. This only needs to
be true in \emph{expectation} over the assignments $\Z$, even if, in practice,
we can show that the inequalities hold for all $\Z$ (cf.
Section~\ref{sec:application}).

We say the interference mechanism is self-exciting because these inequalities
are verified when units benefit from being surrounded by units in treatment. A
successful messaging feature launch is a straightforward example of a
self-exciting process, as is any pricing mechanism that penalizes any treated
bidders and boosts the utility of their competitors.

\subsection{An experiment-of-experiments design}

\begin{figure}
  \centering
  \includegraphics[scale=.6]{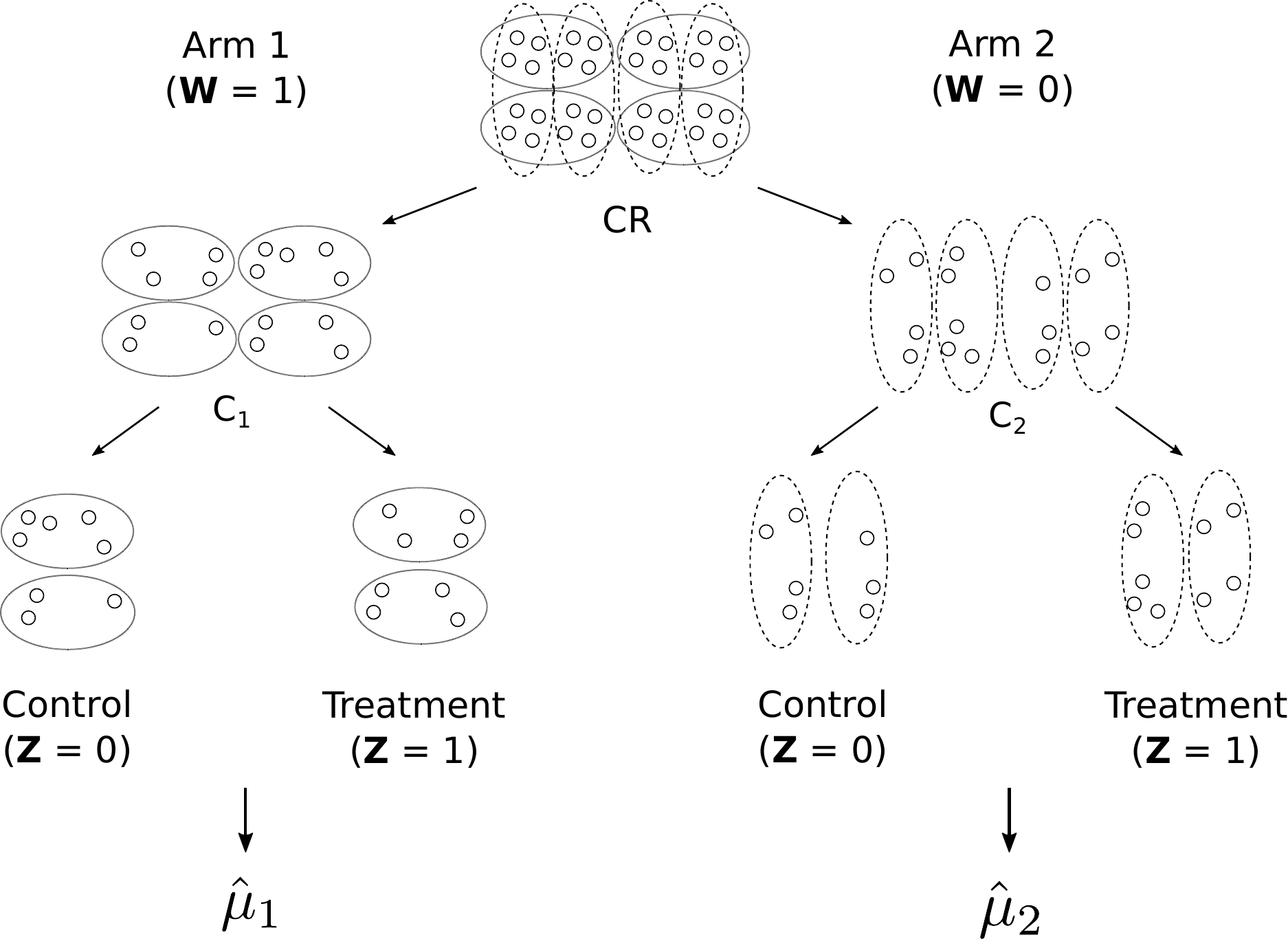}
  \caption{A hierarchical experimental design, which assigns the experimental
  units to one of two cluster-based randomized designs, $C_1$ and $C_2$,
  completely at random (CR).  $\hat \tau^\W_1$ and $\hat \tau^\W_2$ represent
the treatment effect estimates under each design respectively.}\label{fig:hier}
\end{figure}

Under monotonicity, Proposition~\ref{prop:usefulness} states that we can
determine the least-biased of two $\calP$-increasing or $\calP$-decreasing
cluster-based designs, without knowledge of the estimand, by comparing the
expectation of their estimates. However, only one cluster-based design can ever
be applied to the set of experimental units in its entirety, and the comparison
of $\E_{\Z \sim \calC_1} [\hat \tau]$ with $\E_{\Z \sim \calC_2}[\hat \tau]$
cannot be done directly.

This resembles the fundamental problem of causal inference, which
states that units cannot be placed both in treatment and control buckets, and is
solved through randomization. Inspired by~\citep{saveski2017detecting,
pouget2017testing}, we suggest to randomly assign different units to either
clustering algorithm, resulting in a 2-step hierarchical randomized design. The
procedure, described in pseudo-code in Algorithm~\ref{alg:hier}, is as follows:

\begin{itemize}
  \item Assign units completely at random to two design buckets, one for each
    clustering algorithm. Let $\W \in {\{1, 2\}}^N$ be the vector representing
    that assignment.
  \item Within each design bucket, cluster the remaining units together
    according to the appropriate partition: if $W_i = W_j = k$ \emph{and}
    $C_k(i) = C_k(j)$, then $i$ and $j$ belong to the same cluster in design
    bucket $k \in \{1, 2\}$. The resulting partitions are $C_1^\W$ and $C_2^\W$.
  \item Within each design bucket, assign the resulting clusters to treatment
    and control. Let $\Z$ be the resulting assignment vector. This is
    possible because no unit belongs to both $\calC_1^\W$ and $\calC_2^\W$.
\end{itemize}

\begin{algorithm}
\SetAlgoNoLine%
  \KwIn{Partitions $\calC_1,~\calC_2$ of the $N$ units into $M_1,~M_2$
  clusters.}
  \KwOut{$\Z \in {\{0,1\}}^N$ encoding the assignment of each unit to a
  treatment or control bucket.}
\caption{Experiment of experiments design}
Choose $\W \in {\{1, 2\}}^N$ uniformly at random, encoding the
assignment of units to design arms $1$ and $2$\;
\For{$k \in \{1, 2\}$}{%
  Let $C_k^\W$ be the clustering on $\{i\in [1,N]: W_i = k\}$ such that
  $C_k^\W(i) = C_k^\W(j)~\text{iff}~C_k(i) = C_k(j)$\;
  Assign units in treatment arm $k$ to treatment and
  control with a $\calC_k^\W$-cluster-based design\;}
\Return~the resulting assignment vector $\Z$\;
\label{alg:hier}
\end{algorithm}
Algorithm~\ref{alg:hier} provides us with two estimates, $\hat{\tau}^\W_1$ and
$\hat{\tau}^\W_2$, of the causal effect, one from each design arm.  The
resulting clusterings $\calC_1^\W$ and $\calC_2^\W$ may be unbalanced.  This is
of minor importance as the HT estimator (cf.  Eq.~\ref{eq:HT}) is unbiased
(under SUTVA) for unbalanced clusterings, and balancedness is required only to
control its variance. In practice, $\calC_1$ and $\calC_2$ are not required to
have the same number of clusters, but we expect the clusters sizes to be large
enough for each cluster to have at least one unit in each design arm after the
first stage with high probability.

From the comparison of $\hat \tau^\W_1$ and $\hat{\tau}^\W_2$, we seek
to order $\E_{\Z \sim \calC_1}[\hat{\tau}_1]$ and $\E_{\Z \sim
  \calC_2}[\hat{\tau}_2]$.  Under arbitrary interference structures,
these proxy estimates are not guaranteed to have the same ordering,
the key condition for Proposition~\ref{prop:usefulness}.  Intuitively,
$\hat{\tau}^\W_1$ and $\hat{\tau}^\W_2$ represent the treatment effect
estimates for two ``weakened'' versions of each partitioning $\calC_1$
and $\calC_2$.  This is where a completely randomized assignment helps. Because
the assignment of units to design arms is done completely at random, it affects
each partitioning in the same way, and we expect the ordering to stay the same.
For the linear model of interference in Prop.~\ref{prop:linear_monotone}, we
have:
\begin{property}
\label{prop:transitivity}
An interference mechanism is said to be $\calP'$-transitive if $\forall
  (\calC_1, \calC_2) \in \calP'^2$,
  \begin{align*}
    \E_{\W, \Z \sim \calC_1^\W}\left[ \hat \tau^\W_1 \right]
    \leq \E_{\W, \Z \sim \calC_2^\W} \left[  \hat \tau^\W_2\right]
    \Leftrightarrow \E_{\Z \sim \calC_1}[ \hat \tau ]  \leq \E_{\Z \sim
    \calC_2}[ \hat \tau ]
  \end{align*}
\end{property}

As a sanity check, we can also confirm that the property holds for SUTVA\@.~The
property can also be shown for the linear interference mechanisms introduced in
Prop.~\ref{prop:linear_monotone}:
\begin{proposition}\label{prop:linear_transitive}
  Under SUTVA, it holds that
  \begin{equation*}
  \E_{\W, \Z \sim \calC_k^\W} \left[\hat \tau^\W_k \right] =
  \E_{\Z \sim \calC_k}[\hat \tau] = \tau.
\end{equation*}
  Hence, the no-interference case is
  trivially $\calP$-transitive.  Furthermore, the linear model of interference
  in Prop.~\ref{prop:linear_monotone} is $\calP$-transitive if the same number
  of units is assigned to each design arm: $\sum [W_i = 1] = \frac{N}{2}$.
\end{proposition}
A full proof can be found in Section~\ref{sec:proofs}. For more complex
mechanisms of interference, as is the case for reserve price experiments, we
use simulations to confirm the intuition that transitivity holds. See
Section~\ref{sec:experimental} for more details.

As is common with A/B tests, we do not have access to the expectation of our
estimators, and rely on approximations to the variance, such as Neymann's
variance estimator. In order to meaningfully compare the
estimates we obtain, we must apply our method of choice to determine when 
their ordering is significant. For example, we can make a normal
approximation to the distribution of the estimates--- using Neymann's 
estimator to upper-bound the variance ---to estimate the probability that one
estimate is greater than the other with a certain significance level:
\begin{proposition}\label{prop:statistical_test}
 For $k \in \{1, 2\}$, recall the definition of the Neymannian variance
  estimator for cluster-based randomized designs:
  \begin{equation}\label{eq:neymann}
  \hat \sigma_k^\W = \frac{M_k}{N_k} \left(\frac{\hat S_{k, t}}{M_{k,t}} +
    \frac{\hat S_{k, c}}{M_{k,c}} \right),
 \end{equation}
 where $M_k$ (resp. $N_k$) is the number of clusters (resp.~units) in
 $\calC^\W_k$, $\hat S_{k, t} = var\{Y'_{j,k} : z_j = 1\}$ and $\hat S_{k, c} =
 var\{Y'_{j,k} : z_j = 0\}$, and $Y'_{j, k} = \sum_{\calC^\W_k(i) = j} Y_i$.
  Assume that the interference mechanism is transitive and $\calP'$-increasing,
  such that $(\calC_1, \calC_2) \in \calP'^2$. If $\alpha$ is the level of
  significance chosen, we state that $\calC_1$ is a significantly better
  clustering than $\calC_2$ if and only if
  \begin{equation*}
    \Phi\left(\frac{\hat \tau_1^\W - \hat \tau_2^\W}{\sqrt{\hat \sigma_1^\W +
    \hat \sigma_2^\W}}\right) < \alpha,
  \end{equation*}
  where $\Phi$ is the cdf of the normal distribution.
\end{proposition}
A similar reasoning applies to $\calP'$-decreasing mechanisms.  If the Gaussian
approximation is not appropriate, the distribution of the estimators can equally
be approximated by a bootstrap analysis, or a more sophisticated model-based
imputation method~\citep{imbens2015causal}. More details can be found in
Section~\ref{sec:proofs}.

%\begin{align*}
%    \hat \tau^\W_k & =  \frac{M_1}{N_1 M_{T,1}} \sum_{w_j = k} z_j \sum_{i \in
%\calC^\W_{1, j}} Y_i(\Z) \\
%  & \qquad -  \frac{M_1}{N_1 M_{C,1}} \sum_{w_j = k} (1 - z_j)
%\sum_{i \in \calC^\W_{1, j}} Y_i(\Z)
%\end{align*}

%% file: application.tex
\section{Application to reserve price experiments}
\label{sec:application}

Online advertising exchanges provide an interface for bidders to participate in
a set of auctions for advertising online. These ads can appear within the
company's own content, in a social feed, below a search query, or on the webpage
of an affiliated publisher. These auctions provide the vast majority of revenue
to these platforms, and are thus the subject of experimentation and
optimization.
% To understand the impact of these modifications to
%the advertising ecosystem, these p
Platforms run experiments and monitor different metrics including of revenue
and estimates of bidders' welfare. One such welfare metric is the sum of the
bids of advertisers, and another metric is the sum of estimated utility of
bidders via another utility estimator.

One possible parameter subject to optimization is the method of determining reserve
prices. Online marketplaces can choose to implement a reserve price, which sets
the minimum bid required for a bid to be valid and compete with others. It may
vary from bidder to bidder, and from auction to auction.  A higher reserve may
improve revenue, but if it is too high, then too many bids are discarded and ad
opportunities can go unsold.
% a bidder's bid might be
%discarded altogether from the auction if the reserve price exceeds her
%value for winning the auction.

Modifications to a reserve price rule are prime examples of experiments where
SUTVA does not hold.  A change in reserve price to one bidder affects the
bidding problem facing another bidder,  even when her reserve is unchanged
(e.g., reducing competition when the reserve to the first bidder is higher).
%
% which
%other competing bidders may react to. Even in the case of rational bidders
%without budget constraints ---where a change in the reserve price does not
%induce a change in bidding strategy, and solely determines whose bids ultimately
%take part in the auction --- interference persists.
%
Although we ignore them here, budget constraints are another factor--- if a
budget-constrained bidder faces higher reserve prices, then she may adjust her
bids to re-optimize return on investment.
% %
%k%eep within her budget limit
%strategy as a result of a change in her reserve prices.
%
%
Working without budget constraints, we establish conditions under
which the resulting interference mechanism within reserve price
experiments is monotone, both in the case of a single-item second
price auction setting and in the Vickrey-Clarke-Groves auction setting
for positional ads. See~\citep{varian2014vcg} for a reference.

\subsection{Single-item second price auctions}

We consider a single-item second-price auction with $N$ bidders $B =
{\{B_i\}}_{i \in N}$: the highest bidder wins the auction and is
charged the maximum of her reserve price and the second-highest bid.
The second price auction is truthful (bidding true values is a
dominant-strategy equilibrium), and we will assume that the bidders
are rational.

Consider two reserve price mechanisms ${(r_i)}_{i \in B}$ (control) and
${(r'_i)}_{i\in B}$ (treatment). Suppose that the reserve price mechanism
corresponding to treatment always sets a higher reserve price than the reserve
price mechanism corresponding to control: $\forall i, r'_i > r_i$.  By symmetry,
the following argumentation would also work if the treatment and control labels
were switched.

We suppose the bidders have unobserved values $(v_i)$ for winning the auction.
We randomly assign bidders to either the treatment or control reserve price
mechanism, with $\Z$ the resulting assignment.  The chosen metric of interest is
a bidder's utility, denoted by $Y_i(\Z)$. For a second-price auction, $Y_i = 0$
if bidder $i$ does not win the auction, and $Y_i = v_i - p$ when she wins the
auction and pays price $p$. The bidder welfare of an auction is the sum
of each bidder's utility, $\sum_i Y_i(\Z)$,
and  the estimand is given by:
\begin{equation*}
S = \sum_i Y_i(\vec 1) - \sum_i Y_i(\vec 0)
\end{equation*}

Tthe reserve price experiment for second price auctions verifies the
self-excitation property (cf.  Prop.~\ref{prop:more}). The idea is
that assigning a unit to the intervention can only make them less
competitive by discarding their bid from the auction. Thus, the higher
the number of treated units, the lower the competition for the
remaining bidders, and the higher their utility.

\begin{theorem}\label{thm:second_price}
Consider a set of rational agents with no budget-constraints.  Let the outcome
  of interest be each agent's welfare. The interference mechanism of a reserve
  price experiment, assigning treated units to a higher personalized reserve
  price, for a single-item second-price auction is self-exciting, and thus
  monotone.
\end{theorem}

\begin{proof}
Consider bidder i's outcome under $\Z = \vec 0$ and under any assignment $\Z'$
  such that $Z_i = 0$. There are three possible cases:
\begin{itemize}
    \item Bidder $i$ wins the auction in neither assignment. Her utility is
        therefore constant.
    \item Bidder $i$ wins the auction in only one assignment. It must be that
        bidder $i$ wins under $\Z'$ but not $\Z$. Her utility is $0$ under $\Z$
        and greater than $0$ under $\Z'$.
    \item Bidder $i$ wins the auction under both assignments. If the second
        highest bid is the same under both assignments, bidder $i$'s utility is
        constant.  Otherwise, the second highest bid under $\Z'$ can only be
        lower than the second highest bid under $\Z$. Thus bidder $i$'s payment
        is lower and her utility is higher under assignment $\Z'$ than under
        assignment $\Z$.
\end{itemize}
By symmetry, we reach a similar conclusion when comparing assignments
$\Z = \vec 1$ and any assignment $\Z'$ such that $Z'_i = 1$.
\end{proof}

It follows that the reserve price experiment is $\calP$-increasing,
and any cluster-based randomized design underestimates the 
bidder welfare estimand.

\subsection{Positional ad auctions}

\begin{figure}
\centering
\includegraphics[scale=.5]{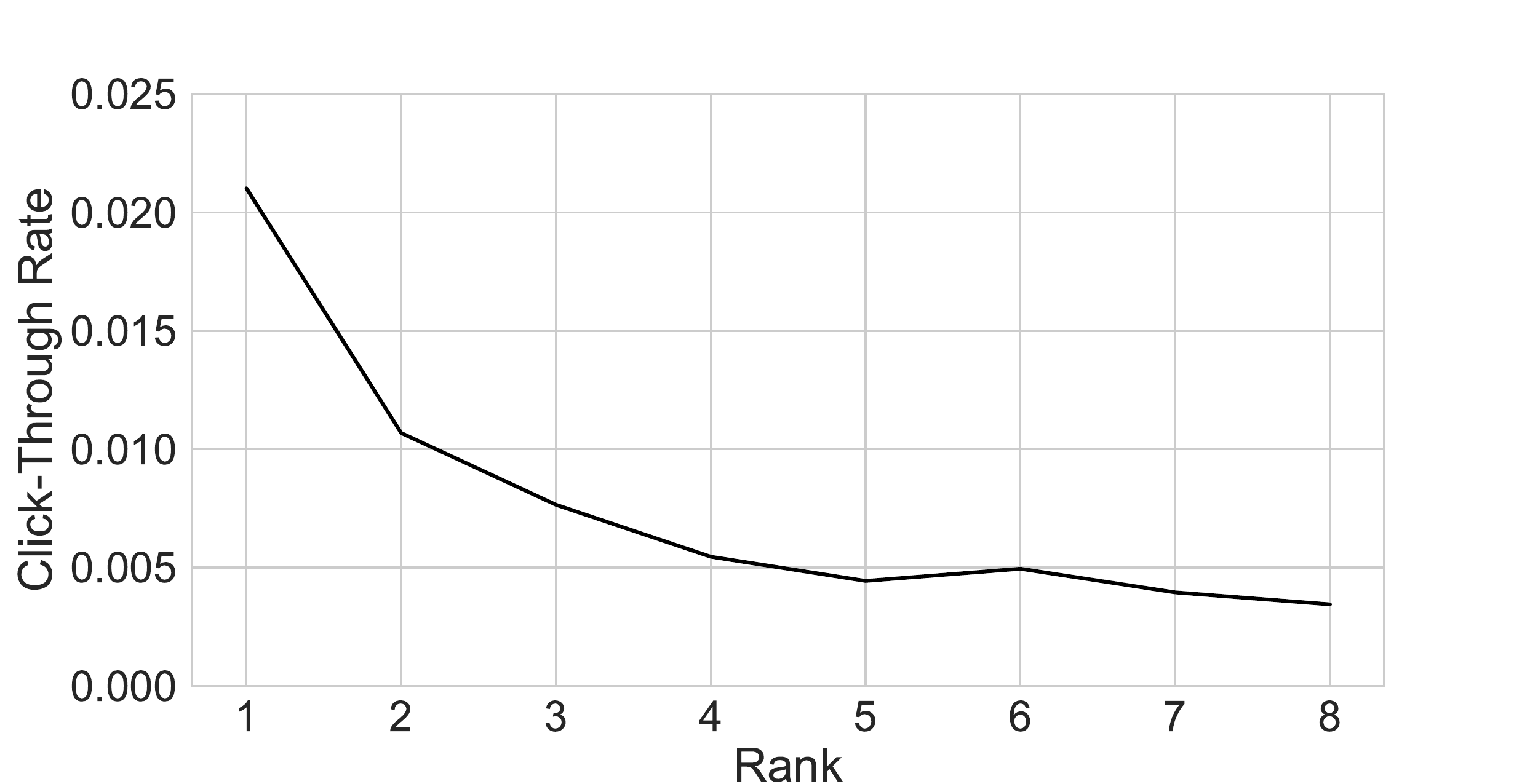}
  \caption{The average click-through rate (CTR)  observed in the \emph{Yahoo!
  Search Auction} dataset, described in Section~\ref{sec:experimental}, can be
  observed to be an approximately decreasing and convex function of the slot
  rank. The confidence intervals were too small to be meaningfully reported in
  the figure.}\label{fig:convexity}
\end{figure}

In practice, ad auctions are also multi-item, used for selling more
than one ad position on a user's view. We now extend the previous
results to a multi-item setting, with $m$ items (or ``slots''). We
assume the common positional ad setting, where each slot has an
inherent click-through rate $pos_j$, which we can suppose is ordered:
$pos_1 > pos_2 > \cdots > pos_m$~\citep{varian2007position}.  Each
bidder $i$ is only ever allocated at most one item, with value $v_i$
for getting a \emph{click}. As a result, bidder $i$'s utility for
winning slot $j$ is $v_i \cdot pos_j - p_i$, where $p_i$ is the
required payment. We assume for simplicity that all bidders have the
same ad quality, and thus the same click-through rate for a given ad
slot.

The Vickrey-Clarke-Groves (VCG) auction  takes place in two parts.  First, a
value-maximising allocation is chosen (based on bids). Here, the highest bids
win the highest slots. Bidders are then charged the externality they impose on
all other bidders. In other words, assuming that bidder $k$ obtains the $k^{th}$
slot, bidder $k$ pays:
\begin{equation*}
  p_k = \sum_{j=k+1}^m (pos_{j-1} - pos_j)\cdot v_j \cdot \mathbbm{1}_{[v_j \geq
  r_j]}
\end{equation*}
where $r_j$ is the reserve imposed on bidder $j$ with value $v_j$.  We can prove
that the self-excitation property holds under a convexity assumption.
\begin{theorem}\label{thm:vcg}
Consider a set of rational agents with no budget-constraints.  Let the outcome
  of interest be each agent's welfare. The interference mechanism of a reserve
  price experiment, assigning treated units to a higher personalized reserve
  price, for a VCG auction in the positional ad setting with no quality effects
  is self-exciting, and thus monotone if the click-through rate function $pos: i
  \mapsto pos_i$ is convex: \begin{equation*} \forall i > j,~pos_{i+1} - pos_i
  \leq pos_{j+1} - pos_j, \end{equation*}
\end{theorem}
This convexity assumption is verified empirically in the literature and in the
Yahoo!  auction dataset\footnote{Our own dataset could potentially suffer from
endogeneity, where weaker bidders are consistently assigned to lower slots. The
assumption is, however, supported elsewhere in the
literature~\citep{brooks2004atlas, richardson2007predicting}.}
introduced in Section~\ref{sec:experimental} (cf.
Figure~\ref{fig:convexity}). The intuition behind the proof is similar: the
greater the number of my competitors are treated, the fewer are able to compete,
and thus the higher my utility. We prove this through a case analysis.
Let $r^\Z_k$ be the reserve that bidder $k$ faces under assignment vector $Z$:
$r^\Z_k = r_k$ if $Z_k = 0$ and $r'_k$ otherwise.

\begin{proof}
    Consider the outcomes of bidder $i$ and $j$ under $\Z$ and $\Z'$ such that
    for all $k \neq j$, $Z_k = Z_k'$, $Z_i = Z_i' = 0$, and $Z_j = 0 < Z_j' =
    1$. By transitivity, if we can show $Y_i(Z) \leq Y_i(Z')$, then it follows
    that $Y_i(\vec 0) \leq \E_\calC[Y_i(\Z) : Z_i = 0]$. There are three
    possible cases:
    \begin{itemize}
      \item The allocation of bidders to slots does not change and thus prices
        do not change. Bidder i's utility is constant.
      \item Bidder $i$ is allocated to slot $i$ for both $\Z$ and $\Z'$
        assignments, but bidder $j$'s ($j < i$) bid is discarded  when $j$ is
        treated ($Z'$): $r_j' > v_j > r_j$. The difference of bidder $i$'s
        outcome under the two treatment assignments is:
          $Y_i(\Z) - Y_i(\Z') = - \sum_{k \geq j} (pos_{k -1} - pos_k) (v_k
          \mathbbm{1}_{v_k > r^\Z_k} - v_{k+1} \mathbbm{1}_{v_{k+1} >
          r^\Z_{k+1}})$.
        This quantity is always negative, hence $Y_i(\Z) \leq Y_i(\Z')$.
      \item Bidder $j$'s ($j < i$) bid is discarded when $j$ is treated and thus
        bidder $i$ is allocated to slot $i-1$.  In that case, bidder $i$'s
        utility under $\Z$ is:
          $Y_i(\Z) = pos_i v_i - \sum_{k \geq i+1} (pos_{k-1} - pos_k) v_k
          \mathbbm{1}_{v_k > r^Z_k}$. The same bidder $i's$ utility under $\Z'$
          is: $Y_i(\Z') =  pos_{i-1} v_i - \sum_{k \geq i+1} (pos_{k-2} - pos_k)
          v_k \mathbbm{1}_{v_k > r^Z_k}$.

        It follows that the difference of bidder $i$'s outcomes is equal to:
         \begin{align*}
           Y_i(\Z)  - Y_i(\Z')  & =   (pos_i - pos_{i-1}) v_i \\
           & \quad - \sum_{k \geq i+1} (pos_{k-2} + pos_k - 2
          pos_{k-1}) v_k,
         \end{align*}
        where the $\mathbbm{1}_{v_k > r^Z_k}$ terms are implicit. Note that each
        individual term of the sum is positive by convexity, such that $Y_i(\Z)
        \leq Y_i(\Z')$.
    \end{itemize}
\end{proof}

%\paragraph{Generalized Second Price.}
%
%The Generalized Second Price (GSP) auction is a simpler mechanism but is not in
%general a truthful one. The bids are ranked in decreasing order and assigned to
%the corresponding slot: Bid $1$ - the highest bid - is assigned to slot $1$,
%with the greatest click-through rate, and so on. Each bidder is then charged the
%maximum of the reserve or the next bidder's bid. \textbf{TODO: what conclusion
%if the auction is not truthful??}

%% file: experimental.tex
\section{Experimental Data and Validation}
\label{sec:experimental}

In this section, we validate our design strategy for comparing two
given graph partitions for the purpose of experimentation under
interference to an advertising auction dataset.
%
%, we can validate the approach through simulation.
%Having explored a theoretical application of our approach to reserve price
%experiments in Section~\ref{sec:application},
For this purpose,  we make use of  a Yahoo! auction
dataset.
% to construct a compelling auction simulation.
%

\subsection{The Yahoo! Search Auction dataset}

{\small
\begin{table}
  \centering
  \begin{tabular}{l r}
  \begin{tabular}{lrl}
    Per keyphrase \\
    \midrule
     nbr of bids & min & 1 \\
     & median & 2 \\
     & max & 7041 \\
     bid value & min & $.3$\textcent\\
     & median & $66$\textcent\\
     & max & \$$320$\\
     impressions & min & 1 \\
     & median & 3 \\
     & max & $5 \cdot 10^6$ \\
     clicks & min & 0 \\
     & $cdf(1)$ & $91.4$ \\
     & max & 7041 \\
  \end{tabular} &
  \begin{tabular}{lrl}
    Per bidder \\
    \midrule
     nbr of bids & min & 1 \\
     & median & 9 \\
     & max & $2.1 \cdot 10^4$ \\
     bid value & min & $.5$\textcent\\
     & median & $60$\textcent\\
     & max & \$$4700$ \\
     impressions & min & $1$ \\
     & median & $31$ \\
     & max & $1.4 \cdot 10^6$ \\
     clicks & min & $0$ \\
     & $cdf(1)$ & $93.3$ \\
     & max & $1.1 \cdot 10^4$\\
  \end{tabular}
  \end{tabular}
  \caption{Summary statistics for the Yahoo! dataset, aggregated by keyphrase
  or by bidder,  per day for the entire 4 month period. Bid values are given in
  USD unless specified otherwise.  $cdf(1)$ is the value of the cumulative
  distribution function of impressions for a single
  impression.}\label{tab:summary}
\end{table}
}  % font-size

The \emph{Yahoo! Search Marketing Advertiser Bid-Impression-Click data on
competing Keywords} dataset is a publicly-available dataset released by
Yahoo!\footnote{Available for download at
\url{https://webscope.sandbox.yahoo.com/}}, containing bid, impression, click,
and revenue data between advertiser-keyphrase pairs over a period of 4 months.
The advertiser and keyphrase are anonymized, represented as a randomly-chosen
string. A sample line of the dataset is reproduced\footnote{The account ID and
keyword ID's have been shortened for the sake of exposition in this sample line.
The bid value is given in 1/100\textcent.} below:

\begin{table}[h]
  \centering
  \begin{tabular}{ccccccc}
  day & id & rank & keyphrase & bid & impress. & clicks \\
    1 & \texttt{a3d2}& 2 & \texttt{f3e4,j6r3,}\dots & 100.0 & 1.0 & 0.0
  \end{tabular}
  \caption{Sample line in the Yahoo! dataset}
\end{table}

The dataset contains $77,850,272$ bidding activities of $16,268$ different
bidders. There are a total of $75,359$ keywords represented, for a total of
unique $648,515$ keyphrases (or list of keywords). Table~\ref{tab:summary}
contains a series of summary statistics computed over keyphrase-day pairs and
bidder-day pairs, namely the total number of bids, the total bid value, the
total number of impressions, and the total number of clicks per keyword (or per
bidder) and per day.

We can represent the \emph{Yahoo!} dataset by a set of bipartite graphs between
bidders, identified by their \texttt{account\_id}, and the keyphrases. The
\emph{bid} bipartite graph on day $t$ draws a weighted edge of weight $w_{ij}$
between every bidder-keyphrase pair such that bidder $i$ bids $w_{ij}$ on
keyphrase $j$ on day $t$. We can aggregate these graphs over the entire time
period ($4$ months) by summing their edge weights together. We can also consider
the impression, rank, and clicks graphs, where the weight of the edge is given
by the number of impressions, the rank, or the number of clicks respectively
received by bidder $i$ on keyphrase $j$.

The dataset only provides data aggregated at the granularity of a single day,
reporting the average bid and total number of impressions and clicks for each
bidder, keyphrase day triplet. Hence, we define a keyphrase-day pair as
a single auction, where each bidder's bid is set to the reported average bid for
that keyphrase-day pair.  For the sake of simplicity, we will only
consider a setting with the first $4$ ad positions, which account for the
majority of clicks.

\subsection{Simulating a reserve price experiment}
\label{sec:bipartite}

\begin{figure}
  \centering
  {\small
  \begin{tabular}{cc}
    \includegraphics[scale=.5]{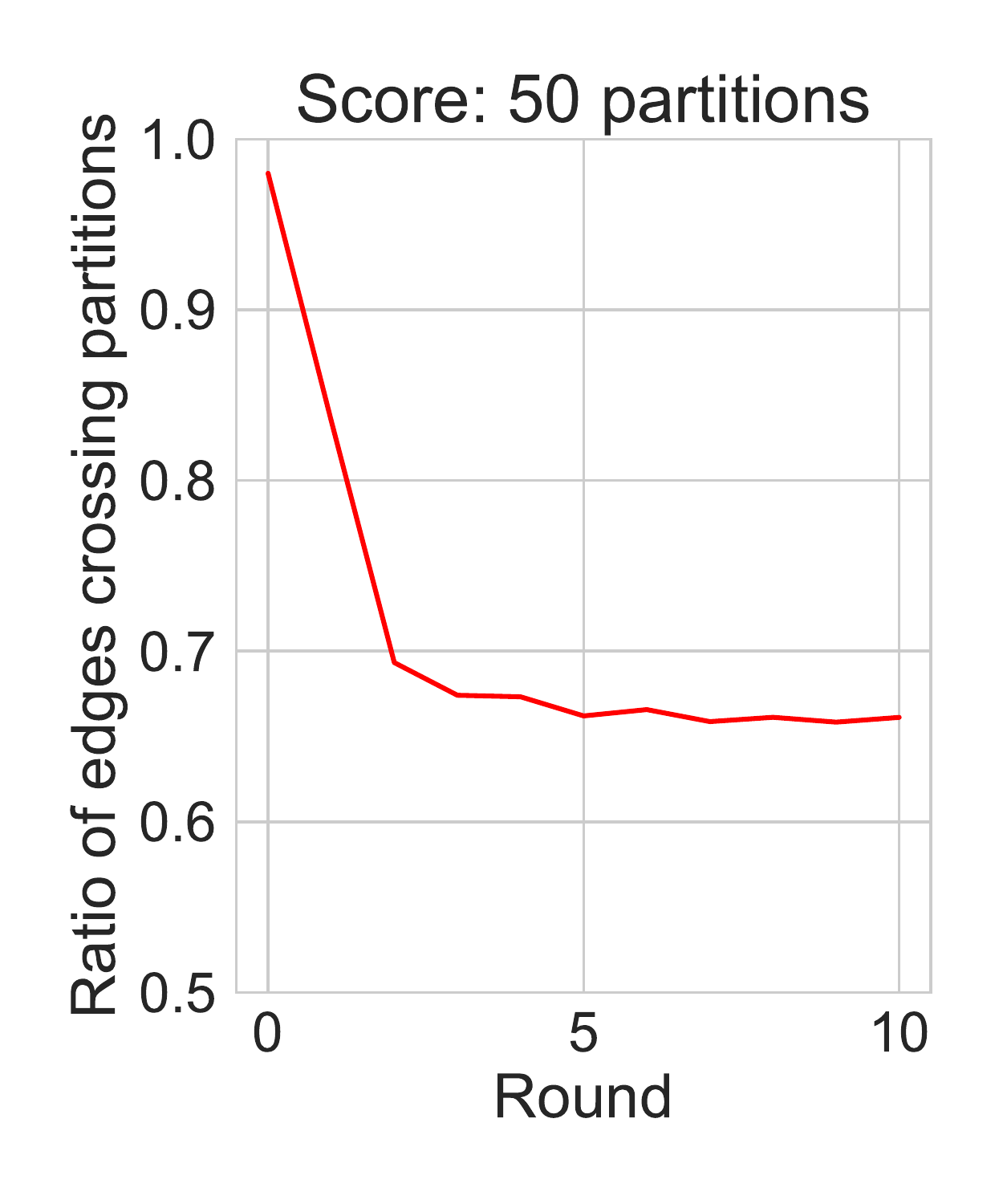} &
    \includegraphics[scale=.5]{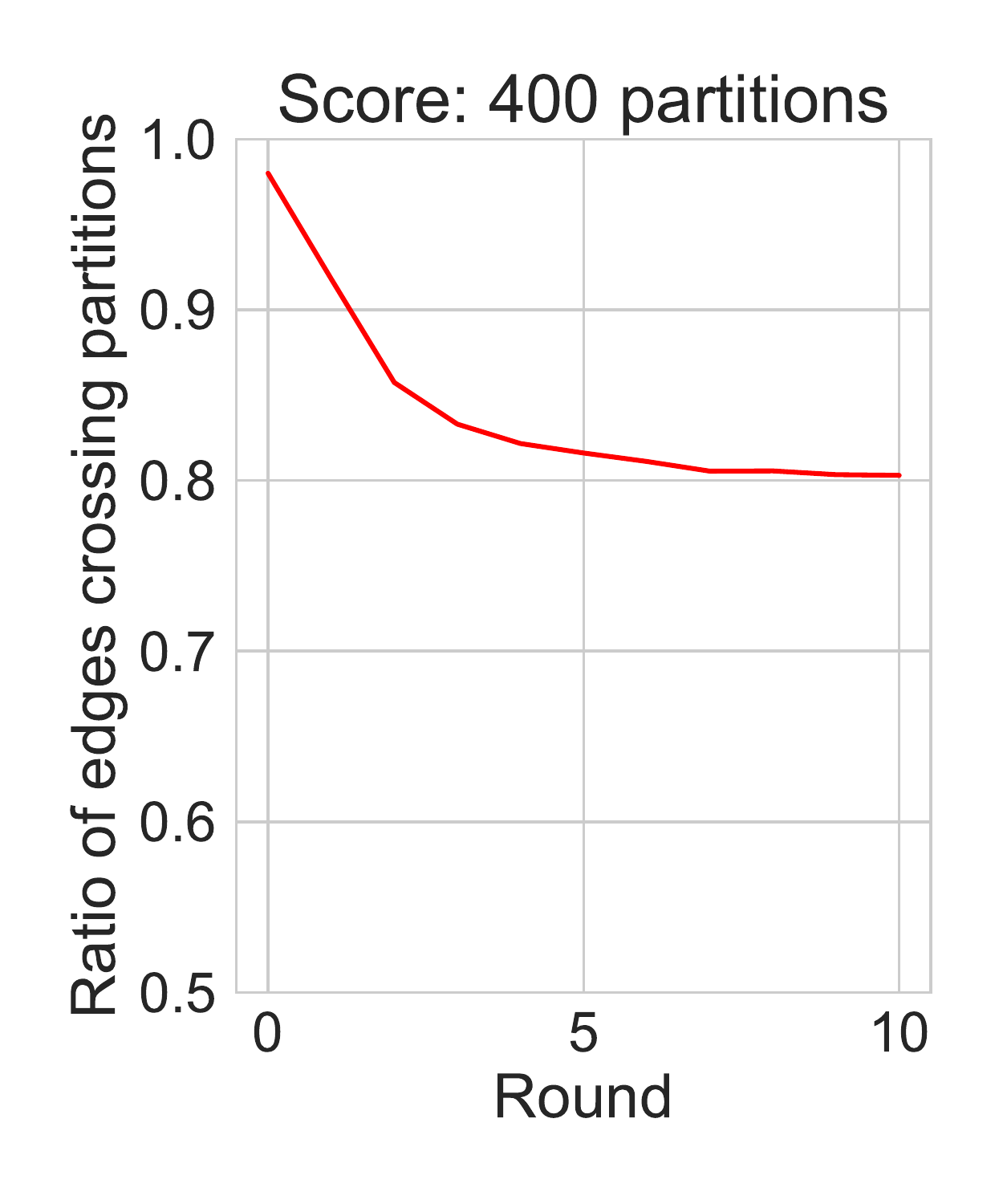} \\
  \end{tabular}
  } % fontsize
  \caption{Weighted ratio of edges across partitions for successive runs of the
  R-LDG algorithm on the weighted bid graph into $50$ partitions and $400$
  partitions respectively.}\label{fig:cut_over_time}
\end{figure}

While the \emph{Yahoo! Search Auction} dataset provides us with a set of
bidders, keyphrases, and the bids, impressions, and clicks that link them, it
does not provide us with an actual intervention on the auction ecosystem. We
must therefore simulate the impact of a change in the reserve price given to
each bidder.

While many possible units of randomization exist for an auction experiment
(keyphrases, bidders, browsers, users, various pairings of these units, etc.),
the reserve price experiment we consider randomizes on bidders. On large auction
platforms, the reserve price might be set
% by a
%sophisticated machine learning algorithm, which attempts to estimate each
%bidder's true value for each auction.
through the application of machine learning methods.  In our context,  we choose
a random non-zero reserve price for each bidder, calibrating the spread of the
distribution such that some bidders will not always be able to match the reserve
price for all auctions.  All bidders assigned to the intervention will face
their non-zero reserve price, fixed for every auction for simplicity.  All
bidders assigned to the control bucket will not face a reserve price.

 Within the same auction for a given
keyphrase, two participating bidders may face distinct reserves and be assigned
to different treatment buckets. A bidder-cluster-based randomized experiment is
thus used to mitigate the possible interference between bidders, our units of
randomization, within a single auction.

To validate our experiment-of-experiments design, we must find candidate
balanced graph partitions to compare, a problem known to be NP-hard --- even
when we slightly relax the balancedness assumption~\citep{andreev2006balanced}.
In the last several years, there has been good progress in developing scalable
distributed balanced partitioning algorithms for graphs with billions of
edges~\citep{TGRV14,ABM16}. These algorithms have enabled practitioners to apply
these large-scale graph mining algorithms for large-scale randomized
experimental studies~\citep{ugander2013balanced,saveski2017detecting,RAKMN16}. Of
the numerous heuristic algorithms for finding such partitions, the {\em
Restreaming Linear Deterministic Greedy} (R-LDG)
algorithm~\citep{nishimura2013restreaming} is a popular choice.  It consists of
repeatedly applying a greedy algorithm, originally proposed
in~\citep{stanton2012streaming}, which assigns each node $u$ to one of $k$
partitions according to the following objective:
\begin{equation*}
  \argmax_{i \in \{1, \dots k\}} |P_i^t \cap N(u) | \left(1 -
  \frac{|P_i^t|}{H_i} \right)
\end{equation*}
where $P_i^t$ is the set of nodes assigned to partition $i$ at step $t$ of the
algorithm, $H_i$ is the maximum capacity of partition $i \in \{1, \dots k\}$,
and $N(u)$ is the set of neigbhors of node $u$ in the graph.

We can apply this clustering algorithm to any of the bipartite graphs introduced
in Section~\ref{sec:bipartite}, aggregated over the entire time period,
resulting in a set of mixed bidder-keyphrase clusters. The bidder-only clusters
are obtained from the previous clustering by simpling removing the keyphrase
nodes from consideration.  The algorithm's objective must be slightly modifed to
accomodate weighted graphs, by replacing $|P_i^t \cap N(u)|$ with $\sum_{i, j}
w_{ij} \mathbbm{1}_{i \in N(u)} \mathbbm{1}_{j \in P_i^t}$. Furthermore, we must
also modify the balance requirement, since only the bidder side of the bipartite
graph clustering is required to be balanced!  We therefore replace $\left(1 -
|P_i^t|/H_i \right)$ with $\left(1 - |P_{i,c}^t|/H_{i,c} \right)$ where $P_{i,
c}^t$ is the set of bidder nodes in partition $P_i^t$ and $H_{i,c}$ is
the maximum number of allowed bidder nodes in partition $P_i^t$.  The
final objective is given by:
\begin{equation*}
  \argmax_{i \in \{1, \dots k\}} \left| \sum_{i \in N(u), j \in P_i^t} w_{ij}
  \right| \left(1 - \frac{|P_{i, c}^t|}{H_{i, c}} \right)
\end{equation*}
Figure~\ref{fig:cut_over_time} plots the proportion of edges cut, weighted by
the bid amount, over consecutive runs of the R-LDG algorithm for $50$ and $100$
clusters. We adopt three main vectors of comparison between
candidates partitions to determine the efficacy of our proposed
experiment-of-experiment design:
\begin{itemize}
  \item \emph{Quality:} comparing partitions of the graph
    that differ in their estimated quality, for example by looking at the number
    of edges cut, for a fixed number of clusters. As an extreme
    example, we will compare a random graph partitioning to a partitioning
    obtained by running the R-LDG algorithm to convergence.
  \item \emph{Number of partitions:} comparing two partitions of the
    graph obtained by running the same clustering algorithm for a different
    number of partitions. As an example, we will consider a R-LDG clustering
    with $10$ clusters and a R-LDG clustering with $400$ clusters.
  \item \emph{Metric:} comparing partitions of the graph that are
    obtained by applying the same algorithm on different bipartite graphs. As an
    example, we will compare a R-LDG clustering of the \emph{bid} graph with an
    R-LDG clustering of the \emph{impressions} graph.
\end{itemize}

The dataset does not provide the budgets of the bidders or their
perceived ad quality, hence we will adopt the same simplifying
assumptions as Section~\ref{sec:application} of no quality effects
between bidders and no budget constraints. Furthermore, we assume bids are
unchanged as a result of the experiment (which would be valid for
rational, non budget-limited bidders).
%
%bidders are
%rational, that they bid truthfully in the Yahoo! dataset, and continue
%to do so in the experiment.
%

\subsection{Validating the empirical optimization}
\label{subsec:validating}

\begin{figure*}
  \centering
  \begin{tabular}{c}
    \includegraphics[scale=.3]{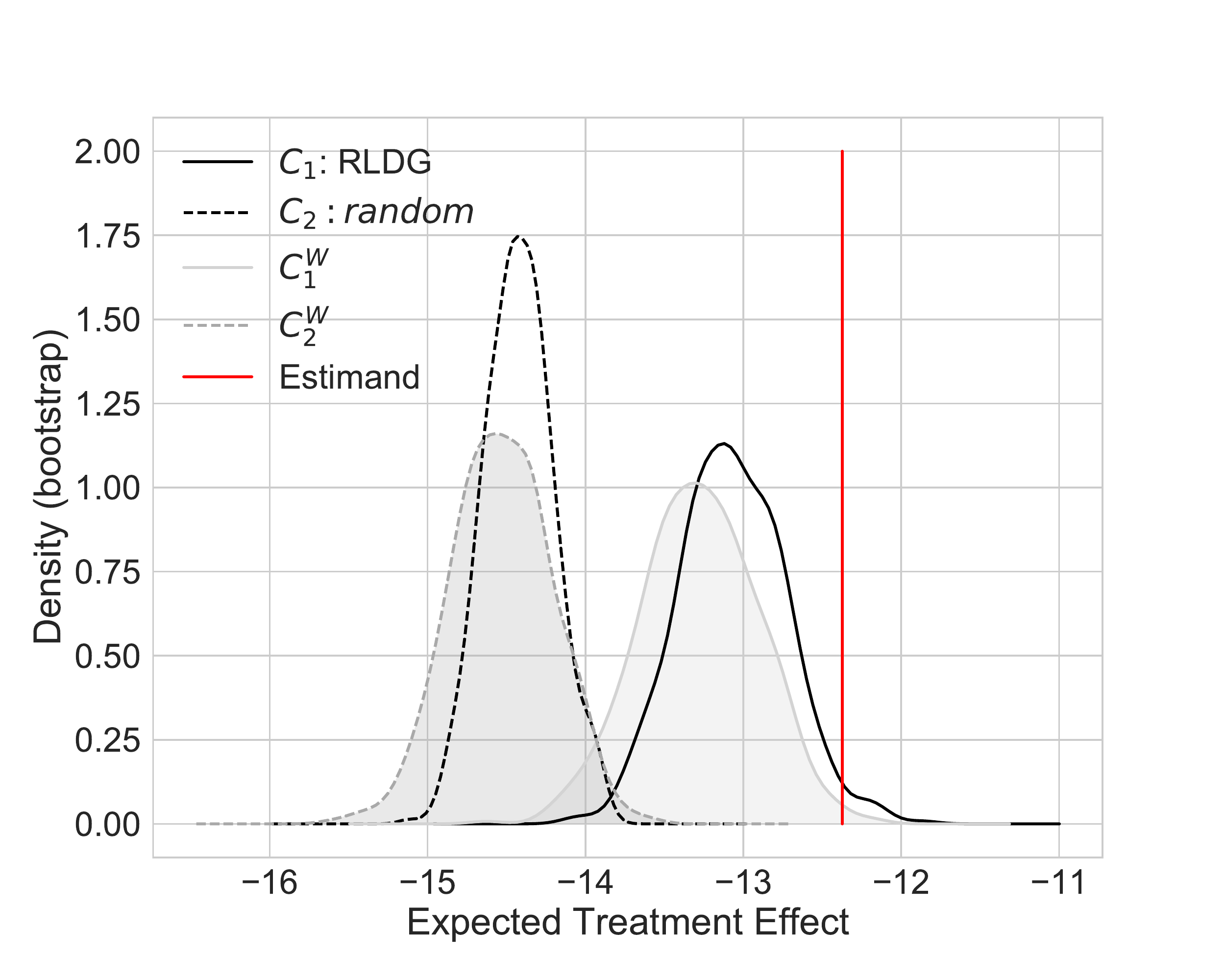} \\
    \includegraphics[scale=.3]{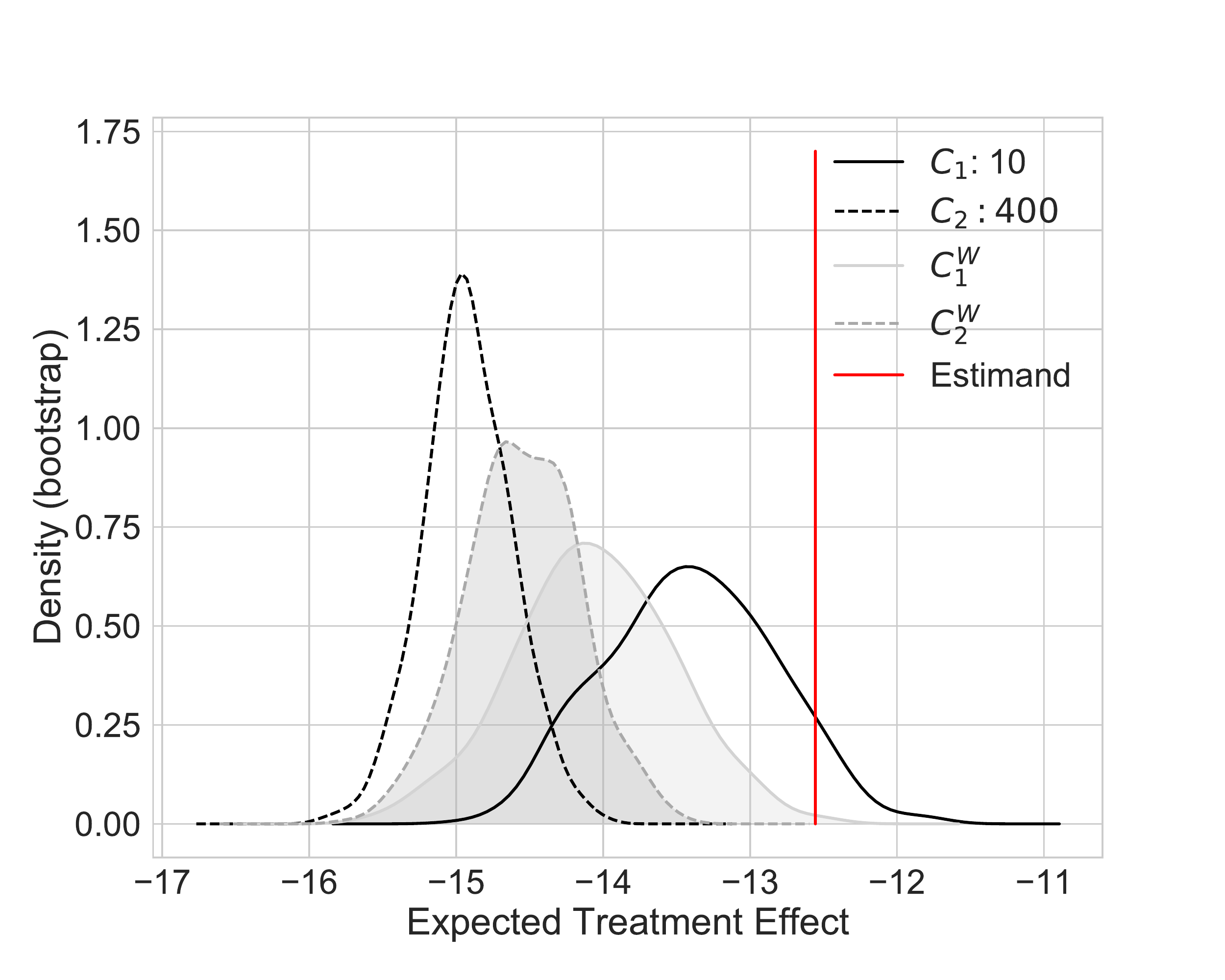}
    \\
    \includegraphics[scale=.3]{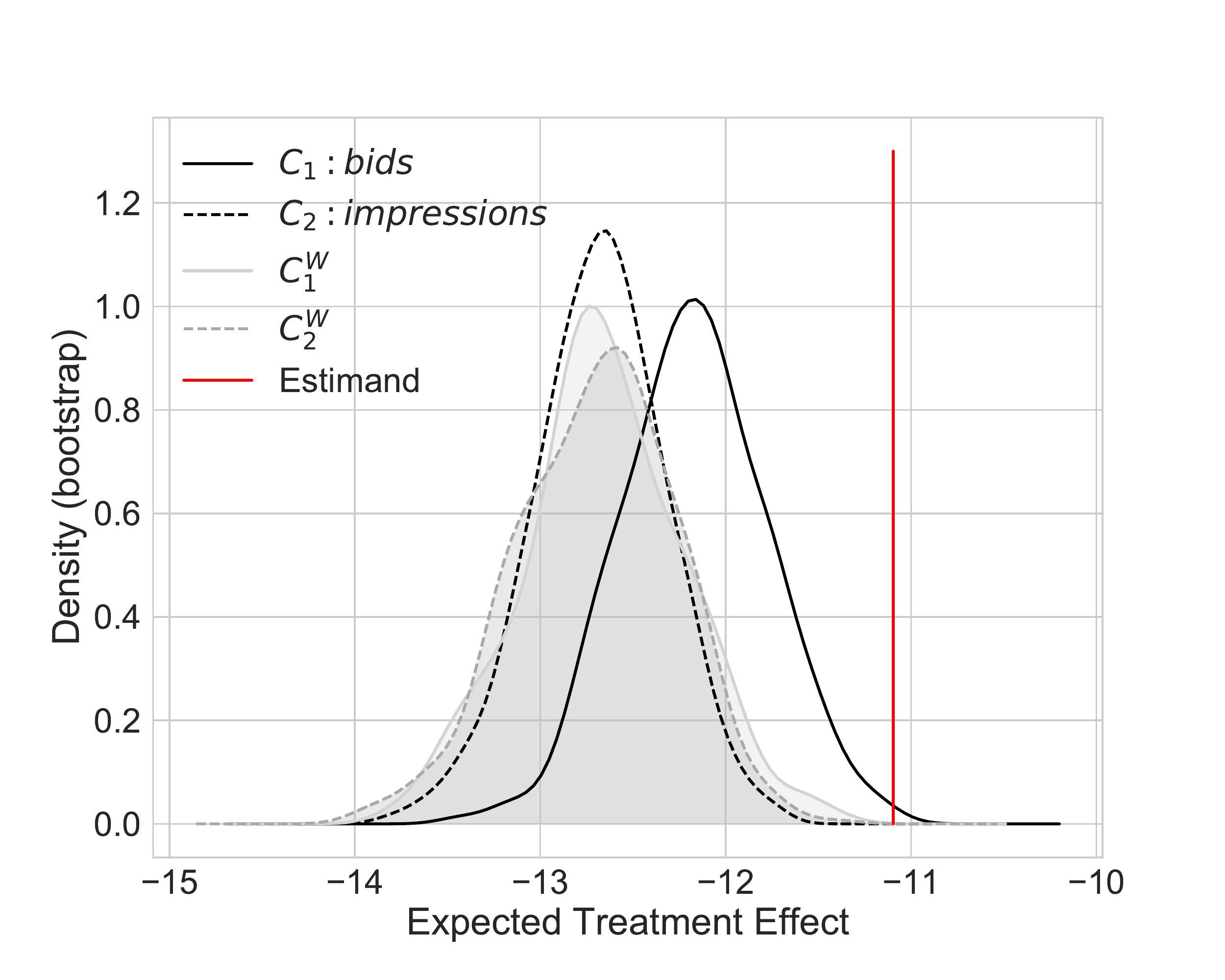} \\
  \end{tabular}
  \caption{Distribution of the expectation of the HT estimator under $\calC_1$
  and $\calC_2$, and the induced clusterings $\calC_1^\W$ and $\calC_2^\W$. The
  red segment represents the total treatment effect estimand.  \emph{(Top)}
  $\calC_1$ is a R-LDG clustering, $\calC_2$ is a random clustering ($M_1 = M_2
  = 50$).  \emph{(Middle)} $\calC_1$ is a R-LDG clustering into $10$ partitions,
  $\calC_2$ is a R-LDG clustering into $400$ partitions. \emph{(Bottom)}
  $\calC_1$ is a R-LDG clustering of the bid graph, whereas $\calC_2$ is a R-LDG
  clustering of the impressions graph. ($M_1 = M_2 = 50$)}\label{fig:results}
\end{figure*}

We first compare a partitioning of the graph obtained by running the modified
R-LDG algorithm (cf.~Section~\ref{sec:bipartite}) against a completely random
balanced partitioning of the graph.  We fix a subset of auctions with few
bidders per auction, in order to showcase the framework and establish the
monotonicity and transitivity properties by allowing a setting for which there
is a clear difference between the two clusterings.  The reduction in cut size
--- measured by the ratio of the weighted sum of edges inter-clusters over the
sum of all edge weights --- over the iterations of the algorithm is shown in
Figure~\ref{fig:cut_over_time}.  While the weighted cut of the graph for a
random partition is around $98\%$, the partition obtained with the R-LDG
algorithm approaches $66\%$ within a few iterations.

We validate the monotonicity assumption, as well as the transitivity
assumption, for reserve price experiments. In Figure~\ref{fig:results} (a), we
plot four distributions as well as the Total Treatment Effect estimand (cf.
Eq.~\ref{eq:tte}), obtained by taking the difference between assigning all units
to a higher reserve price and assigning none. Namely, we plot the distribution
of the HT estimator's expectation (cf. Eq~\ref{eq:HT}) under each cluster-based
design: $\E_{\Z \sim \calC_k}[\hat{\tau}]$ where $k = 1$ for the R-LDG
clustering and $k = 2$ for the random clustering. We also plot the distribution
of the expectation of the experiment-of-experiments (EoE) estimators: $\E_{\W,
\Z \sim \calC^\W_k}[\hat{\tau}_k^\W]$.

We find that they all under-estimate the true treatment effect,
as expected
% which was%
%expected since we showed, in Section~\ref{sec:application}, that reserve price
%experiments were
from the  $\calP$-increasing property. As expected, the HT estimator is more
biased under a random clustering than under the R-LDG clustering. Furthermore,
we find that the property of transitivity holds
(cf.~Eq.~\ref{prop:transitivity}), namely the EoE estimate of the ``random
estimator'' also under-estimates the total treatment effect more severely than
the EoE estimate of the ``R-LDG estimator''.

We repeat the experiment to compare a R-LDG clustering with $10$ partitions
with another R-LDG clustering with $400$ partitions (cf.
Figure~\ref{fig:results} (b)). We find that the clustering with $10$ partitions
is less biased but exhibits higher variance, and that the transitivity property
holds. Finally, in Figure~\ref{fig:results}~(c),  we compare a clustering of
the impressions bipartite graph with a clustering of the bid bipartite graph.
The transitivity property is again verified, and
moreover we see that clustering the bid bipartite
graph may be a better heuristic in this setting, but
the difference in the two
clusterings is very slight.  The code is available for download at
\url{https://jean.pouget-abadie.com/kdd2018code}.

%% file: conclusion.tex
\section{Discussion}

We showed that, under a certain monotonicity assumption, we can determine which
of two clusterings yields the least biased estimator by running an
experiment-of-experiments design. We noted that commonly-studied parametric
models of interference verify this monotonicity assumption. Moreover, we proved
that the interference mechanism resulting from the impact of a reserve price
experiment on social utility is monotone, and hence our framework applies.
Finally, we validated our framework on a simulated reserve price experiment,
grounded in a publicly-available Yahoo! search ad dataset. There are several
questions worth investigating that we did not tackle in this paper. Notably,
while we explored the case of rational bidders participating in positional ad
auctions without budget constraints or quality effects to establish
monotonicity, can these assumptions be relaxed or generalized? What other kinds
of experiments are monotone (or self-exciting)? Is it possible to generalize
Theorem~\ref{thm:vcg} to other Vickrey-Clacke-Groves auctions, \emph{Generalized
Second Price} auctions, or budgeted bidders?  Finally, can the monotonicity
assumption be validated empirically, either through an experimental design or an
observational data study? It seems randomized saturation
designs~\citep{baird2014designing} would be a good place to start for testing
monotonicity experimentally. Finally, our framework relied on the transitivity
of the experiment-of-experiment estimators: namely, that they conserved the
ordering of the expectation of the estimators under each clustering. Whilst we
validated this assumption either
theoretically~(cf.~Prop.~\ref{prop:linear_transitive}) or through
simulation~(cf.~Section.~\ref{subsec:validating}), can we characterize the
clustering-experiment pairs that are transitive and can the assumption be tested
empirically?

%% file: proofs.tex
\section{Proofs}
\label{sec:proofs}

\subsection{Proof of Proposition~\ref{prop:simple_linear_monotone}
and~\ref{prop:linear_monotone}}

Assume that $\forall \Z,~Y_i(\Z) = \alpha_i + \beta_i \cdot Z_i + \gamma_i
\frac{1}{|\N_i|} \sum_{j \in \N_i} Z_j + \epsilon_i$, where $\epsilon_i \sim
\N(0, \sigma^2)$.  Recall the definition of the estimand:
$\tau = \frac{1}{N} \sum_i Y_i(\vec 1) - Y_i(\vec 0)$. Plugging in the expression
for $Y_i(\vec Z)$, we obtain: $\tau = \frac{1}{N} \sum_i \beta_i + \frac{1}{N}
\sum_i \gamma_i$.
The estimator is given by: $\hat \tau = \frac{M}{N} \sum_i
\frac{{(-1)}^{1 - Z_i}}{m_t^{Z_i} m_c^{(1 - Z_i)}} Y_i(\Z)$, where $m_t$
(resp.~$m_c$) is the number of clusters in treatment (resp.~control). Plugging
in the expression for $Y_i(\vec Z)$, we obtain:
\begin{equation*}
  \E_{Z \sim \calC}[\hat \tau] = \frac{1}{N} \sum_i \beta_i + \frac{1}{N} \sum_i
  \gamma_i \left( \frac{|\N_i \cap C(i)|}{|\N_i|}  - \frac{1}{M-1} \frac{|\N_i
  \backslash C(i)|}{|\N_i|} \right)
\end{equation*}
We obtain the desired result by taking the difference between these quantities.
Prop.~\ref{prop:simple_linear_monotone} follows by substituting $\gamma_i =
\gamma$.

\subsection{Proof of Proposition~\ref{prop:more}}
The proposition can be established by rewritting the definition of
$\calP$-increasing interference mechanisms,
\begin{align*}
  \tau - \E_{\Z \sim \calC}[\hat \tau] =\frac{1}{N} \sum_i & \left( Y_i(\vec 1) -
    \E_{\Z\sim\calC}[Y_i(\Z) | z_{C(i)} = 1] \right) \\
    & + \left( \E_{\Z\sim\calC}[Y_i(\Z) | z_{C(i)} = 0] - Y_i(\vec 0) \right),
\end{align*}
such that a sufficient condition of the model to be $\calP$-increasing is for $
Y_i(\vec 1) > \E_{\Z\sim\calC}[Y_i(\Z) | z_{C(i)} = 1]$ and $Y_i(\vec 0) <
\E_{\Z\sim\calC}[Y_i(\Z) | z_{C(i)} = 0]$. If increasing the number of treated
units in that unit's neighborhood increases that unit's outcome --- holding that
unit's treatment assignment constant --- then the two previous inequalities
hold.

\subsection{Proof of Proposition~\ref{prop:linear_transitive}}
Recall that for $k \in \{1, 2\}$, our estimator can be written as:
\begin{equation*}
  \hat{\tau}^\W_k = \frac{M_k}{N_k} \sum_i W_i Y_i(\Z) \frac{{(-1)}^{1 -
  Z_i}}{M^{Z_i}_{k,t} M^{1 - Z_i}_{k,c}},
\end{equation*}
where $M_{k,t}$ (resp.~$M_{k,c}$) is the number of treated (resp.~control)
clusters in design arm $k$ and $N_k$ is the number of units in design arm $k$.
We begin by first considering the no-interference case.
We have that $\E_{Z \sim C_k^\W}[\hat \tau_k | \W ] = \frac{1}{N_k} \sum_i W_i
(Y_i(1) - Y_i(0))$. By the law of iterated expectations, we have $\E_{\W, Z \sim
C_k^\W}[\hat{\tau}^\W_k] = \tau$.

We now consider the linear model suggested in Eq.~\ref{eq:linear}, where
we assume heterogeneous network effects ($\gamma_i$).
From the proof of Proposition~\ref{prop:linear_monotone}, we have that
\begin{equation*}
  \E_{\Z \sim \calC_k^\W}[\hat{\tau}^\W_k | \W] = \bar \beta +
  \frac{M_k}{M_k-1}\frac{1}{N_k} \sum_i W_i \gamma_i \left( \theta_{C_k^\W,
  i} - 1 \right)
\end{equation*}
Note that we have $\E_\W[W_i \theta_{C_k^\W, i}] = \frac{N_k (N_k - 1)}{N (N -
1)} \theta_{C_k, i}$. It follows that, if $M_1 >> 1$,~$M_2 >> 1$, and $N_1 =
N_2 = \frac{N}{2}$,
\begin{align*}
  \E_{\W, \Z \sim \calC_1^\W}[\hat{\tau}^\W_1]- \E_{\W, \Z \sim \calC_2^\W}[
    \hat{\tau}^\W_2] & \approx  \frac{1}{2N} \sum_i \gamma_i \theta_i
  \\
  & \approx \E_{\Z \sim \calC_1}[\hat \tau]- \E_{\Z \sim \calC_2}[\hat \tau]
\end{align*}
We conclude that the linear model of interference is transitive.

\subsection{Discussion for Proposition~\ref{prop:statistical_test}}
Under unspecified models of interference, theoretical bounds on the power of
even the simplest randomized experiment are hard to come by. While the joint
assumption of monotonicity and transitivity allow us to design a sensible test
for detecting the better of two partitions, they are not sufficient to
bound its power without stronger assumptions. We thus rely on simulations, like
the ones run in Section~\ref{sec:experimental}, or theoretical approximations,
like the ones suggested in Prop.~\ref{prop:statistical_test}. It approximates
$\E_{\W, \Z}[\hat{\tau}_k^\W]$, for $k \in \{1, 2\}$ by two
independently-distributed Gaussian variables of mean $\hat{\tau}_k^\W$ and
variance $\hat{\sigma}_k^\W$, given in Eq.~\ref{eq:neymann}.
Their difference therefore has the distribution $\N(\hat{\tau}_1^\W -
\hat{\tau}_2^\W, \hat{\sigma}_1^\W + \hat{\sigma}_2^\W)$.  Recall that Neymann's
variance estimator is an upper-bound of the true variance, under SUTVA, in
expectation over the assignment $\Z$ (cf.~\citep{imbens2015causal}). We prove in
the lemma below that this still holds true for a hierarchical assignment.
\begin{lemma} Under SUTVA, Neymann's variance estimator is an upper-bound in
  expectation of the true variance of the HT estimator:
  \begin{equation*}
    \E_{\W,\Z}[\hat{\sigma}_k^\W] \geq var_{\W, \Z}[\hat{\tau}_k^\W]
  \end{equation*}
\end{lemma}
\vspace{-1.1em}
\begin{proof}
  By Eve's law, $var_{\W, \Z}[\hat{\tau}_k^\W] = \E_\W[ var_{\Z \sim
  \calC_k^\W}[\hat{\tau_k^\W} | \W]] + var_\W[\E_{\Z \sim \calC_k^\W}
  [\hat{\tau}_k^\W]]$. From~\citep{imbens2015causal}, the first term can is equal
  to:
  \begin{equation*}
    \frac{M_k}{N_k} \left(\frac{var(Y'(1))}{M_{k,t}} +
    \frac{var(Y'(0))}{M_{k,c}} - \frac{var(Y'(1) - Y'(0))}{M_k} \right),
  \end{equation*}
  where $Y'_j(Z) = \sum_{i \in \calC_k^\W(j)} Y_i(Z)$, the cluster-level
  outcomes. The second term can be shown to be equal to $\frac{var(Y(1) -
  Y(0))}{N}$.

  Since $\E_{\W, \Z}[\hat{\sigma}_k^2] = \frac{M_k}{N_k}\left(
  \frac{var(Y'(1))}{M_{k,t}} + \frac{var(Y'(0))}{M_{k,c}}\right)$, we must
  prove: $\frac{var(Y'(1) - Y'(0))}{N_k} \geq \frac{var(Y(1) - Y(0))}{N}$. This
  follows from an application of the Cauchy-Schwarz inequality for balanced
  clusters: $\sum_j {(\sum_i Y_i)}^2 \leq \sum_j |C_j| \sum_i Y_i^2$, where
  $C_j$ are the cluster sizes, equal to $\frac{N}{N_k}$ in the balanced case.
\end{proof}

In order to determine the greater of two clusterings, we can perform two
one-sided t-tests.  The Bayesian approach is to compute the posterior
distribution of the difference of the two estimates, using a conjugate Gaussian
prior. In order to assess the impact of assuming the two estimates are
independent Gaussians, we suggest running a sensitivity analysis, by considering
the result of the test for different values of the correlation coefficient.